\documentclass[11pt]{amsart}
\usepackage{amsmath, amsthm, amsfonts, amsbsy, amssymb, upref, enumerate, bigstrut}

\usepackage{epsfig, graphicx}
\usepackage{pdfsync, hyperref}

\newtheorem{thm}{Theorem}[section]
\newtheorem{lmm}[thm]{Lemma}

\newcommand{\ee}{\mathbb{E}}

\newcommand{\ra}{\rightarrow}
\newcommand{\rr}{\mathbb{R}}
\newcommand{\smallavg}[1]{\langle #1 \rangle}

\newcommand{\var}{\mathrm{Var}}

\newcommand{\zz}{\mathbb{Z}}
 
\newcommand{\fpar}[2]{\frac{\partial #1}{\partial #2}}

\newcommand{\mpar}[3]{\frac{\partial^2 #1}{\partial #2 \partial #3}}

\begin{document}
\title[Absence of RSB in random field Ising model]{Absence of replica symmetry breaking in the random field Ising model}
\author{Sourav Chatterjee}
\address{\newline Department of Statistics \newline Stanford University\newline Sequoia Hall, 390 Serra Mall \newline Stanford, CA 94305\newline \newline Email: \textup{\tt souravc@stanford.edu}}
\thanks{Research partially supported by NSF grant DMS-1309618}

\subjclass[2010]{}

\begin{abstract}
It is shown that replica symmetry is not broken in the random field Ising model in any dimension, at any temperature and field strength, except possibly at a measure-zero set of exceptional temperatures and field strengths.
\end{abstract}

\maketitle

\section{Introduction}
Let $\zz^d$ denote the $d$-dimensional integer lattice. For each $n$, let $V_n := \zz^d \cap [1,n]^d$. The random (Gaussian) field Ising model (RFIM) on $V_n$ with free boundary condition is the random probability measure on the set of spin configurations $\{-1,1\}^{V_n}$ that is proportional to
\[
\exp\biggl(\beta \sum_{\langle x y\rangle}\sigma_x \sigma_y + h\sum_{x} g_x\sigma_x\biggr)\, ,
\]
where $\langle x y\rangle$ below the first sum means that we are summing over $x,y$ that are neighbors, $\beta$ and $h$ are positive numbers that we will call {\it inverse temperature} and {\it field strength}, and $g_x$ are independent standard Gaussian random variables. The variables $g_x$ are collectively called the {\it disorder}. 

The RFIM was introduced in a seminal paper of Imry and Ma \cite{imryma75} as a simple example of a disordered system. The model is often showcased as a notable success story of mathematical physics. In the paper \cite{imryma75}, Imry and Ma gave an intuitive argument that the model does not have an ordered phase in dimensions one and two, but does in dimensions three and higher. There was spirited controversy in the theoretical physics community about this claim for more than ten years after the introduction of this model, with arguments coming both for and against. The question was finally settled in two rigorous papers by Bricmont and Kupianen \cite{bk87, bk88} who proved that there is indeed an ordered phase in dimensions three and higher, and two other rigorous papers by Aizenman and Wehr \cite{aw89, aw90} who proved the nonexistence of an ordered phase in dimension two. (The one dimensional case was settled earlier.) For a readable account of these proofs and an up-to-date survey of the literature, see \cite[Chapter 7]{bovier06}.

However, in spite of all this progress, the RFIM has not yet yielded  an exact solution, and there are many aspects that remain unknown. A popular theoretical physics approach to solving statistical models of disordered systems is the so-called `replica method'. There have been a number of attempts to apply this technique to the RFIM by theoretical physicists (e.g.~in~\cite{my92}), with varying degrees of success.  A crucial step in the replica approach is to understand whether the model under investigation exhibits `replica symmetry breaking'. The purpose of this paper is to investigate this question for the RFIM and rigorously prove that replica symmetry does not break in this model. This lends support to the recent finding in the physics literature that the RFIM does not have a spin glass phase \cite{krz10, krsz11} and provides partial justification for the replica symmetric solution of this model. 

It is worth stressing, however, that even from the point of view of the physicists, the replica symmetric solution to the RFIM is not the final answer. There is a general claim, not yet verified rigorously, that the RFIM has a large number (increasing with the system size) of so-called `metastable' states. These metastable states play an important role, but how to identify and count them is still not clear (other than in some simple cases, see e.g.~the case of RFIM on random graphs discussed in~\cite{krz10}).

Let me now briefly explain what is meant by the absence of replica symmetry breaking. Given a fixed realization of the disorder, the probability measure defined above is called the Gibbs measure of the system.  Suppose that two spin configurations $\sigma^1$ and $\sigma^2$ are drawn independently from this random Gibbs measure. These are called {\it replicas}. The {\it overlap} between these two replicas is defined as
\[
R_{1,2} := \frac{1}{|V_n|} \sum_x \sigma_x^1 \sigma_x^2\, .
\]
The system is said to exhibit replica symmetry breaking (e.g., according to Parisi \cite{parisi02}) if the limiting distribution of the random variable $R_{1,2}$, when $n\ra\infty$, has more than one point in its support. The following theorem shows that this does not happen for the RFIM. 
\begin{thm}\label{rsbthm}
There is a set $A \subseteq (0,\infty)^2$ of Lebesgue measure zero such that for any $(\beta, h)\not\in A$, there exists a constant $q_{\beta,h}\in [-1,1]$ such that
\[
\lim_{n\ra\infty} \ee\langle (R_{1,2}-q_{\beta, h})^2\rangle = 0\, ,
\]
where $\langle \cdot \rangle$ denotes averaging with respect to the Gibbs measure and $\ee$ is expectation with respect to the disorder. Moreover for every $\beta$, the set of all $h$ such that $(\beta,h)\in A$ is countable.
\end{thm}
The rest of this paper is devoted to the proof of the above theorem. The proof depends on two key ingredients: the FKG property of the RFIM, and the Ghirlanda-Guerra identities. The proof should extend  easily to any other model that has these two properties.

\section{Proof}
For every $\beta$, $h$ and $n$, define
\[
Z_n(\beta, h) := \sum_{\sigma\in \{-1,1\}^{V_n}}\exp\biggl(\beta \sum_{\langle x y\rangle}\sigma_x \sigma_y + h\sum_{x} g_x\sigma_x\biggr)\, . 
\]
Let $F_n(\beta, h) := \log Z_n(\beta,h)$. Let 
\[
\psi_n(\beta,h) := \frac{F_n(\beta, h)}{|V_n|}
\]
and $p_n(\beta,h) := \ee(\psi_n(\beta,h))$.
\begin{lmm}\label{limlmm}
For every positive $\beta$ and $h$, the limit 
\[
p(\beta,h) := \lim_{n\ra\infty} p_n(\beta,h)
\]
exists and is finite.
\end{lmm}
The proof of this lemma is quite standard, but may be hard to find. For the sake of completeness, a proof is given below.
\begin{proof}
Fix $\beta$ and $h$. Let $Z_n = Z_n(\beta,h)$, $F_n = F_n(\beta,h)$ and $p_n=p_n(\beta,h)$. Take any two integers $m$ and $n$, and let $l:=mn$. Then the vertices of the box $V_l$ may be partitioned into $m^d$ disjoint translates of $V_n$ that we will call `blocks'. Define a new Hamiltonian on $V_l$ by taking the old one and deleting the terms $\sigma_x\sigma_y$ when $x$ and $y$ belong to different blocks. Let $Z_l'$ be the partition function for this new Hamiltonian. Then clearly $Z_l'/Z_l$ is bounded above and below by $e^{\beta L}$ and $e^{-\beta L}$ respectively, where $L$ is the number of deleted bonds. Note that $L\le C(d) n^{d-1}m^d$, where $C(d)$ is a constant that depends on $d$ only. Thus,
\[
|\log Z_l' - \log Z_l|\le C(d) \beta n^{d-1}m^d\,.
\]
But $Z_l'$ factorizes into $m^d$ independent terms, each of which has the same distribution as $Z_n$. Thus, $\ee(\log Z_l')= m^d \ee(F_n)$. By the above inequality, this shows that
\begin{align*}
|p_n - p_l| &= |n^{-d} \ee(F_n) - l^{-d} \ee(F_l)| \\
&= l^{-d}|m^d \ee(F_n)-\ee(F_l)|\le l^{-d} C(d) \beta n^{d-1} m^d = C(d)\beta n^{-1}\, .
\end{align*}
By the symmetry of the situation, $|p_m-p_l|\le C(d)\beta m^{-1}$. Combining, we get $|p_m-p_n|\le C(d)\beta(m^{-1}+n^{-1})$. This shows that the sequence $p_n$ is Cauchy and hence convergent. To show that the limit is finite, simply apply Jensen's inequality while trying to evaluate the expectation in $p_n = n^{-d}\ee(F_n)$. 
\end{proof}

The next lemma is also standard. 
\begin{lmm}\label{convlmm}
The limit $p(\beta, h)$ defined in Lemma \ref{limlmm} is a convex function of $h$ for every fixed $\beta$. The same is true for $F_n$, $\psi_n$ and $p_n$. 
\end{lmm}
\begin{proof}
It is easy to differentiate $F_n$ twice in $h$ and show that the derivative is positive, which implies convexity of $F_n$ as a function of $h$. This shows that $p_n$, and therefore $p$, is convex. 
\end{proof}
Let $A$ be the set of all $(\beta,h)\in (0,\infty)^2$ such that the function $p$ is not differentiable in $h$ at the point $(\beta, h)$. 
\begin{lmm}\label{difflmm}
The set $A$ has Lebesgue measure zero. Moreover, for every $\beta$, the set of all $h$ such that $(\beta,h)\in A$ is countable.
\end{lmm}
\begin{proof}
For fixed $\beta$, $p$ is a convex function of $h$ by Lemma \ref{convlmm}. It is well known (and easy to prove) that the set of points at which a convex function on $(0,\infty)$ is not differentiable must be countable. This proves the second claim of the lemma. For the first, use the second, together with a simple application of Fubini's theorem. 
\end{proof}
The following lemma is also quite standard; as before, a proof is provided for the sake of completeness.
\begin{lmm}\label{varbd}
For any $\beta$, $h$ and $n$, 
\[
\var (F_n(\beta, h)) \le h^2 |V_n|\, .
\]
\end{lmm}
\begin{proof}
Note that by the Poincar\'e inequality for the Gaussian measure \cite[p.~49]{ledoux01}, 
\begin{align*}
\var (F_n(\beta, h))\le \sum_{x\in V_n} \ee\biggl(\fpar{F_n(\beta,h)}{g_x}\biggr)^2 = h^2\sum_{x\in V_n} \ee\smallavg{\sigma_x}^2
\end{align*}
Since $\smallavg{\sigma_x}^2\le 1$ for all $x$, this completes the proof of the lemma. 
\end{proof}
The next lemma records that the RFIM has the FKG property.
\begin{lmm}\label{fkg}
Let $f$ and $g$ be two monotone increasing functions on the configuration space $\{-1,1\}^{V_n}$. Then $\smallavg{fg} \ge \smallavg{f}\smallavg{g}$. In particular, for any $x$ and $y$, $\smallavg{\sigma_x\sigma_y}\ge \smallavg{\sigma_x}\smallavg{\sigma_y}$. 
\end{lmm}
\begin{proof}
It is simple to verify that for any value of the disorder, the Gibbs measure of the RFIM satisfies the FKG lattice condition \cite{fkg71}, which implies positive correlations between increasing functions.  
\end{proof}
The next lemma is the key ingredient in the proof of Theorem \ref{rsbthm} and the main original component of this paper. The proof makes a crucial use of the FKG property of the RFIM.
\begin{lmm}\label{mainlmm}
For any $n$ and any $(\beta,h)\in (0,\infty)^2$, 
\[
\ee(\smallavg{R_{1,2}^2} - \smallavg{R_{1,2}}^2) \le \frac{2\sqrt{2+h^2}}{h\sqrt{|V_n|}}\, .
\]
\end{lmm}
\begin{proof}
Fix $\beta$, $h$ and $n$, and let $F := F_n(\beta,h) - \ee(F_n(\beta,h))$. Let $W$ denote the set of all unordered pairs $\{x,y\}$, where $x,y\in V_n$ and $x\ne y$. For each $w = \{x,y\}\in W$, let $s_w := g_xg_y$, and let
\[
c_w := \ee(s_wF)\, .
\]
Define
\[
S := \sum_{w\in W} c_w s_w\, .
\]
Note that for any two distinct $w,w'\in W$, $\ee(s_ws_{w'}) =0$, and for any $w$, $\ee(s_w^2)=1$. Thus,
\[
\ee(S^2) = \sum_{w} c_w^2 = \sum_w \ee(c_w s_w F) = \ee(SF)\, .
\]
Consequently,
\[
\ee(F-S)^2 = \ee(F^2)-2\ee(FS)+\ee(S^2) = \ee (F^2)-\ee(S^2)\, .
\]
Combining the last two displays and applying Lemma \ref{varbd} gives
\begin{equation}\label{cw}
\sum_w c_w^2\le \ee(F^2)=\var (F_n(\beta, h))\le h^2 |V_n|\, .
\end{equation}
Applying integration by parts, for any $w=\{x,y\}$, 
\begin{align*}
c_w &= \ee(g_x g_y F) = \ee\biggl(\mpar{F}{g_x}{g_y}\biggr) = h^2 \ee(\smallavg{\sigma_x\sigma_y} - \smallavg{\sigma_x}\smallavg{\sigma_y})=: h^2 r_{x,y}\, .
\end{align*}
Therefore by \eqref{cw},  
\begin{align}\label{rxy}
\sum_{x,y\in V_n} r_{x,y}^2 \le  2\sum_{\{x,y\}\in W} r_{x,y}^2  + |V_n| \le \frac{(2+h^2)|V_n|}{h^2}\, .
\end{align}
Now note that
\begin{align*}
\ee(\smallavg{R_{1,2}^2} - \smallavg{R_{1,2}}^2) &= \frac{1}{|V_n|^2} \sum_{x,y} \ee(\smallavg{\sigma_x\sigma_y}^2 - \smallavg{\sigma_x}^2\smallavg{\sigma_y}^2)\\
&\le \frac{1}{|V_n|^2} \sum_{x,y} \ee|(\smallavg{\sigma_x\sigma_y} - \smallavg{\sigma_x}\smallavg{\sigma_y})(\smallavg{\sigma_x\sigma_y} + \smallavg{\sigma_x}\smallavg{\sigma_y})|\\
&\le \frac{2}{|V_n|^2} \sum_{x,y} \ee|\smallavg{\sigma_x\sigma_y} - \smallavg{\sigma_x}\smallavg{\sigma_y}|\, .
\end{align*}
But by the FKG property of the RFIM (Lemma \ref{fkg}), $\smallavg{\sigma_x\sigma_y} -\smallavg{\sigma_x}\smallavg{\sigma_y}\ge 0$ for each $x$ and $y$. Thus, by the above inequality and \eqref{rxy},
\begin{align*}
\ee(\smallavg{R_{1,2}^2} - \smallavg{R_{1,2}}^2) &\le \frac{2}{|V_n|^2} \sum_{x,y}r_{x,y} \\
&\le \frac{2}{|V_n|^2} \biggl(|V_n|^2 \sum_{x,y}r_{x,y}^2\biggr)^{1/2}\le \frac{2\sqrt{2+h^2}}{h\sqrt{|V_n|}}\, .
\end{align*}
This completes the proof of the lemma. 
\end{proof}
Next, define
\[
H_n := \frac{1}{|V_n|}\sum_{x\in V_n} g_x \sigma_x\, .
\]
Recall the set $A$ defined in the paragraph preceding Lemma~\ref{difflmm}. Let $A^c$ denote the complement of $A$ in $\rr^2$. 
\begin{lmm}\label{hnlmm}
For any $(\beta, h)\in A^c$, 
\[
\lim_{n\ra\infty} \ee\smallavg{H_n} = \fpar{p}{h}(\beta,h)
\]
and 
\[
\lim_{n\ra\infty}\ee|\smallavg{H_n}-\ee\smallavg{H_n}| = 0\, .
\]
\end{lmm}
\begin{proof}
Note that
\begin{equation*}\label{diff}
\smallavg{H_n} = \fpar{\psi_n}{h} \ \text{ and } \  \ee\smallavg{H_n} = \fpar{p_n}{h}\, .  
\end{equation*}
Fix some $h' > h>0$. By the convexity of $\psi_n$ (Lemma \ref{convlmm}) and the first identity in the above display, 
\begin{equation}\label{upbd}
\smallavg{H_n} \le \frac{\psi_n(\beta,h')-\psi_n(\beta,h)}{h'-h}. 
\end{equation}
By Lemma \ref{varbd}, 
\begin{equation*}\label{conc}
\ee|\psi_n(\beta,h) - p_n(\beta,h)| \le \sqrt{\var(\psi_n(\beta,h))}\le  \frac{h}{\sqrt{|V_n|}}\,.
\end{equation*}
Therefore, if $p'$ denotes the function $\partial p/\partial h$, then
\begin{align*}
&\ee\biggl| \frac{\psi_n(\beta, h')-\psi_n(\beta, h)}{h'-h} - p'(\beta,h)\biggr| \\
&\le \frac{h + h'}{h'-h}\frac{1}{\sqrt{|V_n|}}  + \frac{|p_n(\beta, h') - p(\beta, h')| + |p_n(\beta, h) - p(\beta,h)|}{h'-h} \\
&\qquad + \biggl|\frac{p(\beta,h') - p(\beta, h)}{h'-h} - p'(\beta,h)\biggr|\, . 
\end{align*}
Since $|V_n|\ra \infty$ and $p_n(\beta,h)\ra p(\beta,h)$ by Lemma \ref{limlmm}, this gives
\begin{align*}
\limsup_{n\ra \infty} \ee\biggl| \frac{\psi_n(\beta, h')-\psi_n(\beta,h)}{h'-h} - p'(\beta,h)\biggr|&\le \biggl|\frac{p(\beta, h') - p(\beta,h)}{h'-h} - p'(\beta,h)\biggr|.
\end{align*}
Combining with \eqref{upbd}, we get
\begin{align*}
\limsup_{n\ra \infty} \ee(\smallavg{H_n}-p'(\beta, h))^+ &\le \biggl|\frac{p(\beta, h') - p(\beta, h)}{h'-h} - p'(\beta,h)\biggr|\, ,
\end{align*}
where $x^+$ denotes the positive part of a real number $x$. Since this bound holds for any $h' > h$ and $p$ is differentiable at $h$, 
\[
\lim_{n\ra \infty} \ee(\smallavg{H_n}-p'(\beta,h))^+ = 0\,. 
\]
Similarly, considering $h' < h$ and repeating the steps, we can show that the limit of the negative part is zero as well. Thus,
\[
\lim_{n\ra \infty} \ee|\smallavg{H_n}-p'(\beta,h)| = 0\,. 
\]
By Jensen's inequality, this gives
\[
\lim_{n\ra \infty}  |\ee\smallavg{H_n}-p'(\beta,h)| = 0\,. 
\]
The proof is now easily completed by combining the last two inequalities.
\end{proof}

\begin{lmm}\label{intparts}
For any $\beta$, $h$ and $n$,
\[
\sum_{x, y} \biggl(\ee\biggl(\frac{\partial^4 F_n}{\partial g_x^2 \partial g_{y}^2}\biggr)\biggr)^2\le 24h^2|V_n|\, .
\]
\end{lmm}
\begin{proof}
This proof is very similar to that of Lemma \ref{mainlmm}. Let $F:= F_n(\beta,h)-\ee(F_n(\beta,h))$. Let $U$ be the set of all unordered pairs $\{x,y\}$, where $x,y\in V_n$. For any $u=\{x,y\}\in U$, define
\[
t_u :=
\begin{cases}
\frac{1}{2}(g_x^2-1)(g_y^2-1) &\text{ if } x\ne y\, ,\\
\frac{1}{\sqrt{24}}(g_x^4-6g_x^2+3) &\text{ if } x=y\, .
\end{cases}
\]
Define $d_u := \ee(t_uF)$  
and 
\[
T := \sum_{u\in U} d_u t_u\, .
\]
An easy computation shows that for any two distinct $u,u'\in U$, $\ee(t_ut_{u'}) =0$, and for any $u$, $\ee(t_u^2)=1$. Thus,
\[
\ee(T^2) = \sum_{u} d_u^2 = \sum_u \ee(d_u t_u F) = \ee(TF)\, .
\]
Consequently,
\[
\ee(F-T)^2 = \ee(F^2)-2\ee(FT)+\ee(T^2) = \ee (F^2)-\ee(T^2)\, .
\]
Combining the last two displays and applying Lemma \ref{varbd} gives
\begin{equation*}
\sum_u d_u^2\le \ee(F^2)=\var (F_n(\beta, h))\le h^2 |V_n|\, .
\end{equation*}
Now, integration by parts shows that if $u=\{x,y\}$ for some $x\ne y$,  
\begin{align*}
d_u &= \frac{1}{2}\,\ee\biggl(\frac{\partial^4 F_n}{\partial g_x^2 \partial g_y^2}\biggr)\, ,
\end{align*}
and for $u=\{x,x\}$, 
\begin{align*}
d_u &= \frac{1}{\sqrt{24}}\,\ee\biggl(\frac{\partial^4 F_n}{\partial g_x^4}\biggr)\, .
\end{align*}
This completes the proof of the lemma. 
\end{proof}
\begin{lmm}\label{hnlmm2}
For any $(\beta,h)\in A^c$,
\[
\lim_{n\ra\infty}\ee\smallavg{|H_n - \ee\smallavg{H_n}|} =0\, .
\]
\end{lmm}
\begin{proof}
Integration by parts gives
\begin{align*}
\ee(\smallavg{H_n^2} - \smallavg{H_n}^2)&= \frac{1}{|V_n|^2}\sum_{x,y\in V_n} \ee(g_x g_y (\smallavg{\sigma_x \sigma_y} - \smallavg{\sigma_x}\smallavg{\sigma_y}))\\
&= \frac{1}{|V_n|^2}\sum_{x,y} \ee\biggl(\mpar{}{g_x}{g_{y}} (\smallavg{\sigma_x \sigma_y} - \smallavg{\sigma_x}\smallavg{\sigma_y})\biggr) \\
&\qquad + \frac{1}{|V_n|^2}\sum_{x} \ee(1-\smallavg{\sigma_x}^2)\, . 
\end{align*}
Now note that
\[
\smallavg{\sigma_x \sigma_y} - \smallavg{\sigma_x}\smallavg{\sigma_y} = \frac{1}{h^2}\mpar{F_n}{g_x}{g_{y}}\, .
\]
Thus, by the Cauchy-Schwarz inequality and Lemma~\ref{intparts}, 
\begin{align*}
\ee(\smallavg{H_n^2} - \smallavg{H_n}^2) &\le \frac{1}{h^2|V_n|^2}\sum_{x,y} \ee\biggl(\frac{\partial^4 F_n}{\partial g_x^2 \partial g_{y}^2}\biggr) + \frac{1}{|V_n|}\\
&\le \frac{1}{h^2|V_n|^2 } \biggl(|V_n|^2\sum_{x, y} \biggl(\ee\biggl(\frac{\partial^4 F_n}{\partial g_x^2 \partial g_{y}^2}\biggr)\biggr)^2\biggr)^{1/2} + \frac{1}{|V_n|}\\
&\le \frac{\sqrt{24}}{h\sqrt{|V_n|}} + \frac{1}{|V_n|}\, . 
\end{align*}
Combining this with the inequality
\[
\ee\smallavg{|H_n- \smallavg{H_n}|} \le \sqrt{\ee\smallavg{(H_n- \smallavg{H_n})^2}} = \sqrt{\ee(\smallavg{H_n^2} - \smallavg{H_n}^2)}
\]
implies that $\ee\smallavg{|H_n- \smallavg{H_n}|} \ra 0$. Since $(\beta, h)\in A^c$, we can now apply Lemma \ref{hnlmm} to complete the proof. 
\end{proof}
Take any integer $k\ge 2$ and let $\sigma^1,\ldots, \sigma^k, \sigma^{k+1}$ denote $k+1$ spin configurations drawn independently from the Gibbs measure.  Let $R_{i,j}$ denote the overlap between $\sigma^i$ and $\sigma^j$. Let $f$ be a function of these overlaps, that takes value in $[-1,1]$ and does not change with $n$. The following result shows that the RFIM satisfies the {\it Ghirlanda-Guerra identities} \cite{ghirlandaguerra98} (also sometimes called the {\it Aizenman-Contucci identities} \cite{aizenmancontucci98}) at almost all $(\beta,h)$. 
\begin{lmm}\label{ggid}
If $(\beta,h)\in A^c$, and $f$ is as above, then
\[
\lim_{n\ra\infty} \biggl(\ee\smallavg{f R_{1,k+1}} - \frac{1}{k}  \ee\smallavg{f} \ee\smallavg{R_{1,2}}- \frac{1}{k}\sum_{i=2}^k \ee\smallavg{fR_{1,i}}\biggr)=0\, . 
\]
\end{lmm}
Given a result like Lemma \ref{hnlmm2}, the derivation of the Ghirlanda-Guerra identities is quite standard \cite{talagrand03}. As usual, details are presented below for the sake of completeness.
\begin{proof}
Integration by parts gives 
\begin{align*}
\ee\smallavg{H_n(\sigma^1) f} &= \frac{1}{|V_n|} \sum_x \ee(g_x \smallavg{\sigma_x^1 f})\\
&=  \frac{1}{|V_n|} \sum_x \ee\biggl(\fpar{\smallavg{\sigma_x^1 f}}{g_x}\biggr)\\
&= h\sum_{i=1}^k\ee\smallavg{fR_{1,i}} - hk \ee\smallavg{f R_{1,k+1}}\, .
\end{align*}
In particular,
\begin{equation}\label{ehn}
\ee\smallavg{H_n} = h(1-\ee\smallavg{R_{1,2}})\, .
\end{equation}
Therefore,
\begin{align*}
&\ee\smallavg{H_n(\sigma^1) f} - \ee\smallavg{H_n}\ee\smallavg{f} \\
&= h\ee\smallavg{R_{1,2}} \ee\smallavg{f} + h\sum_{i=2}^k \ee\smallavg{fR_{1,i}} - hk \ee\smallavg{f R_{1,k+1}}\, .
\end{align*}
To complete the proof, note that by Lemma \ref{hnlmm}, the left-hand side in the above display tends to zero as $n\ra\infty$.
\end{proof}
\begin{proof}[Proof of Theorem \ref{rsbthm}]
We now have all the ingredients necessary for completing the proof of Theorem \ref{rsbthm}. The main ingredient is Lemma \ref{mainlmm}. This needs to be combined with the Ghirlanda-Guerra identities and a nice trick of Guerra \cite{guerra96} which also appears in Talagrand~\cite{talagrand03}. 

Choosing $k=2$ and $f = R_{1,2}$ in Lemma \ref{ggid} gives 
\begin{align}\label{gg1}
\lim_{n\ra\infty} \biggl(\ee\smallavg{R_{1,2}R_{1,3}} - \frac{1}{2}  (\ee\smallavg{R_{1,2}})^2- \frac{1}{2}\ee\smallavg{R_{1,2}^2}\biggr)=0. 
\end{align}
Choosing $k=3$ and $f=R_{2,3}$ gives
\begin{align}\label{gg2}
\lim_{n\ra\infty} \biggl(\ee\smallavg{R_{2,3} R_{1,4}} - \frac{1}{3}  (\ee\smallavg{R_{1,2}})^2- \frac{1}{3}\ee\smallavg{R_{2,3}R_{1,2}} - \frac{1}{3}\ee\smallavg{R_{2,3}R_{1,3}}\biggr)=0\, . 
\end{align}
Now, since by symmetry
\[
\ee\smallavg{R_{2,3}R_{1,2}} = \ee\smallavg{R_{2,3}R_{1,3}} = \ee\smallavg{R_{1,3}R_{1,2}}\, ,
\]
we can use \eqref{gg1} to evaluate the last two terms in \eqref{gg2}, to get
\[
\lim_{n\ra\infty} \biggl(\ee\smallavg{R_{2,3} R_{1,4}} - \frac{2}{3}  (\ee\smallavg{R_{1,2}})^2- \frac{1}{3}\ee\smallavg{R_{1,2}^2}\biggr)=0\, . 
\]
Since $\smallavg{R_{2,3}R_{1,4}} = \smallavg{R_{1,2}}^2$, we can now combine the above display with Lemma \ref{mainlmm} to have
\begin{align*}
\lim_{n\ra\infty} (\ee\smallavg{R_{1,2}^2} - (\ee\smallavg{R_{1,2}})^2)=0\, .
\end{align*}
But the term on the left-hand side is nothing but $\ee\smallavg{(R_{1,2}-\ee\smallavg{R_{1,2}})^2}$. To complete the proof, note that by \eqref{ehn} and Lemma \ref{hnlmm}, $\lim_{n\ra\infty}\ee\smallavg{R_{1,2}}$ exists.
\end{proof}

\vskip.2in
\noindent {\bf Acknowledgments.} I thank Giorgio Parisi, Federico Ricci-Tersenghi and the anonymous referees for helpful comments and suggestions.

\end{document}